\documentclass{article}
\usepackage{microtype}
\usepackage{graphicx}
\usepackage{booktabs,arydshln} 
\makeatletter
\def\adl@drawiv#1#2#3{%
        \hskip.5\tabcolsep
        \xleaders#3{#2.5\@tempdimb #1{1}#2.5\@tempdimb}%
                #2\z@ plus1fil minus1fil\relax
        \hskip.5\tabcolsep}
\newcommand{\cdashlinelr}[1]{%
  \noalign{\vskip\aboverulesep
           \global\let\@dashdrawstore\adl@draw
           \global\let\adl@draw\adl@drawiv}
  \cdashline{#1}
  \noalign{\global\let\adl@draw\@dashdrawstore
           \vskip\belowrulesep}}
\makeatother

\usepackage{hyperref}
\usepackage[round]{natbib}

\usepackage[margin=1in]{geometry}
\usepackage{amsmath, amsthm, amssymb, graphicx}
\DeclareGraphicsExtensions{.jpg,.pdf,.eps}

\usepackage{amssymb}

\usepackage{mathtools}

\usepackage{xfrac}

\usepackage{mathrsfs}
\usepackage{amscd}
\usepackage{dsfont}
\usepackage{tikz}

\usepackage{enumerate}

\usepackage[labelfont=bf]{subcaption}
\usepackage{float}

\theoremstyle{plain}
\newtheorem{thm}{Theorem}[section] 

\newtheorem{prop}[thm]{Proposition}

\newtheorem{lemma}[thm]{Lemma}

\theoremstyle{definition}
\newtheorem{defn}[thm]{Definition} 

\newtheorem{example}[thm]{Example}
\newtheorem{remark}[thm]{Remark}

\usepackage{listings}
\usepackage{color} 

\usepackage{verbatim}

\newcommand{\mscr}[1]{\mathscr{#1}}

\newcommand{\twid}[1]{\widetilde{#1}}

\newcommand{\ZZ}{\mathbb{Z}}
\newcommand{\RR}{\mathbb{R}}

\newcommand{\mcH}{\mathcal{H}}
\newcommand{\mcF}{\mathscr{F}}
\newcommand{\mcX}{\mathcal{X}}
\newcommand{\mcY}{\mathcal{Y}}

\newcommand{\de}{\delta}
\newcommand{\ep}{\epsilon}

\newcommand{\la}{\lambda}

\newcommand{\Sa}{\Sigma}
\newcommand{\De}{\Delta}

\usepackage{mathtools}

\renewcommand{\l}{\left}
\renewcommand{\r}{\right}

\newcommand{\defeq}{\vcentcolon=}

\newcommand{\red}{\color{red}}

\newcommand{\data}[2]{#1_1\ldots,#1_{#2}}
\newcommand{\iid}{\overset{\text{iid}}{\sim}}

\DeclareMathOperator{\vspan}{span}
\DeclareMathOperator{\tr}{tr}

\newcommand{\ul}{\underline}

\title{Benefits and Pitfalls of the Exponential Mechanism with Applications to Hilbert Spaces and Functional PCA }
\author{Jordan Awan \quad Ana Kenney \quad Matthew Reimherr\quad Aleksandra Slavkovi\'c
  \\Department of Statistics, Pennsylvania State University, University Park, PA}

\date{}

\begin{document}
\maketitle

\setcounter{page}{1}
\begin{abstract}


The exponential mechanism is a fundamental tool of Differential Privacy (DP) due to its strong privacy guarantees and flexibility. 
We study its extension to settings with summaries based on infinite dimensional outputs such as with functional data analysis, shape analysis, and nonparametric statistics. 
We show that one can design the mechanism with respect to a specific base measure over the output space, such as a Guassian process.
We provide a positive result that establishes a Central Limit Theorem for the exponential mechanism quite broadly. 
We also provide an apparent negative result, showing that the  magnitude of the noise introduced for privacy is asymptotically non-negligible relative to the statistical estimation error. We develop an $\ep$-DP mechanism for functional principal component analysis, applicable in separable Hilbert spaces. We demonstrate its performance via simulations and applications to two datasets.
\end{abstract}


\section{Introduction}
Data privacy and security have become increasingly critical to society as we continue to collect troves of highly individualized data.  In the last decade, we have seen the emergence of new tools and perspectives on data privacy 
such as {\it Differential Privacy} (DP), introduced by \citet{Dwork2006:Sensitivity}, which provides a rigorous and interpretable definition of privacy.  Within the DP framework, numerous tools have been developed that achieve DP in a variety of applications and contexts, such as  empirical risk minimization
\citep{Chaudhuri2011:DPERM,Kifer2012:PrivateCERM}, linear and logistic regression 
\citep{Chaudhuri2009:Logistic,Zhang2012:FunctionalMR,Yu2014,Sheffet2017,Awan2018:Structure}, 
hypothesis testing \citep{Vu2009,Wang2015:PrivacyFree,Gaboardi2016,Awan2018:Binomial,Canonne2018}, network data \citep{Karwa2016:Sharing,Karwa2016:Inference}, and density estimation \citep{Wasserman2010:StatisticalFDP}, to name a few.

One of the most flexible and convenient DP tools is the \textit{exponential mechanism}, introduced by \citet{McSherry2007}, which often fits in naturally with estimation techniques from statistics and machine learning.  Many estimation procedures can be described as maximizing a particular objective or utility function:
\[
\hat b = \arg\max_{b \in \mathcal Y} \xi(b),
\qquad \text{where } \xi:\mathcal Y \to \RR,
\]
or, equivalently, minimizing a loss function such as least squares or the negative log-likelihood.  The exponential mechanism provides a sanitized version of $\hat b$ by using the objective function directly to add noise. The sanitized estimate, $\tilde b$, is drawn from a density, $f(b)$, that is proportional to
\[
f(b) \propto \exp\left\{ \frac{\epsilon}{2\Delta} \xi(b) \right\},
\]
where $\Delta$ captures the {\it sensitivity} of the objective function to small perturbations in the data, and $\epsilon$ is the desired {\it privacy budget} (details in Sections \ref{s:dp} and \ref{s:exponential}).  The idea behind this mechanism is to assign higher density values to regions with higher utility.  The constant $\Delta/\epsilon$ adjusts the spread of the density; as the sensitivity increases or as the privacy budget decreases (meaning a decreased disclosure risk), the variability of $\tilde b$ increases.  A major advantage of such an approach is its use of the objective function from the non-private estimate, $\hat b$, which naturally promotes perturbations with higher utility and discourages those with poor utility.   

In this paper we study the exponential mechanism, especially as it pertains to functional data analysis, shape analysis, and nonparametric statistics, where one has a (potentially) infinite dimensional output.  We show that the exponential mechanism can be applied in such settings, but requires a specified base measure over the output space $\mcY$.  
We propose using a Gaussian process as the base measure, as these distributions are well studied and easy to implement.   
We derive a Central Limit Theorem (CLT) for the exponential mechanism quite broadly, however, this result also implies that   the magnitude of the noise introduced for privacy is of the same order as the statistical estimation error. 
In particular, we show that in most natural settings the exponential mechanism does not add an asymptotically negligible noise, even in finite dimensions.  

Using our approach, we develop an $\ep$-DP mechanism for functional principal component analysis (FPCA), which extends the method of \citet{Chaudhuri2013} to separable Hilbert spaces. We show that a Gaussian process base measure enables us to modify the Gibbs sampling procedure of \citet{Chaudhuri2013} to this functional setting. We illustrate the performance of our private FPCA mechansim through simulations, and apply our mechanism to both the Berkeley growth study from the \texttt{fda} package \cite{fdapackage} and the Diffusion Tensory Imaging (DTI) dataset from the \texttt{refund} package \cite{refundpackage}.


{\bfseries Related Work: }This work most directly builds off of  \citet{Hall2013} and \citet{Mirshani2007}, which develop the first techniques for producing fully functional releases under DP. Another work in this direction is \citet{Alda2017}, in which they use Bernstein polynomial approximations to release functions. Recently, \citet{Smith2018} applied the techniques of \citet{Hall2013} to privatize gaussian process regression. In their setup, they assume that the predictors are public knowledge, and use this information to carefully tailor the sanitization noise. 

There have been a few accuracy bounds regarding exponential mechanism, which can be found in Section 3.4 of \citet{Dwork2014:AFD}. However, these results bound the loss in terms of the objective function, rather than in terms of the private release. \citet{Wasserman2010:StatisticalFDP} also develop some accuracy bounds for the exponential mechanism, focusing on mean and density estimation. They show that in the mean estimation problem, the exponential mechanism introduces $O(1/\sqrt n)$ noise. Our asymptotic analysis of the exponential mechanism agrees in this setting, and extends this result to a large class of objective functions. 

Our application to FPCA extends the private PCA method proposed in  \citet{Chaudhuri2013}. There have been other approaches to private multivariate PCA. \citet{Blum2005} were one of the first to develop a DP procedure for principal components, which is a postprocessing of a noisy covariance matrix. \citet{Dwork2014} follow the same approach and develop bounds for this algorithm; they also develop an online algorithm for private PCA. \citet{Jiang2013} modify this approach by both introducing noise in the covariance matrix as well as to the projection. \citet{Imtiaz2016} also add noise to the covariance matrix, but use a Wishart distribution rather than normal or Laplace noise.

{\bfseries Organization: }
In Section \ref{s:dp}, we review the necessary background of Differential Privacy. In Section \ref{s:exponential}, we recall the exponential mechanism and give asymptotic results for the performance of the exponential mechanism in both finite and infinite dimensional settings. In Section \ref{s:pca} we show how the exponential mechanism can be applied to produce Functional Principal Components, and in Section \ref{s:sampling} we give a Gibbs sampler for this mechanism. In Section \ref{s:numerical}, we study the performance of the private principal components on both simulated data and on the Berkeley and DTI datasets. 
Finally, we give our concluding remarks in Section \ref{s:conclusions}.




\section{Differential Privacy}\label{s:dp}

In this section we provide a brief overview of differential privacy (DP). 
Throughout, we let $\mcX$ denote an arbitrary set, which represents a particular population, and let $\mcX^n$ be the $n$-fold Cartesian product, which represents the collection of all possible samples that could be observed.  
We begin by defining the {\it Hamming Distance} between two databases.  
\begin{defn}[Hamming Distance]\label{HammingDistance}
  The bivariate function $\de:\mathcal X^n \times \mathcal X^n \rightarrow \ZZ$, which maps $\de(X,Y) \defeq \#\{i \mid X_i \neq Y_i\}$, is called the \emph{Hamming Distance} on $\mathcal X^n$.
\end{defn}
\noindent It is easy to verify that $\de$ is a metric on $\mathcal X^n$. If $\de(X,Y)=1$ we call $X$ and $Y$ \emph{adjacent}.  

Since we are focused on infinite dimensional objects, we define {\it Differential Privacy} broadly for any statistical summary.  In particular, suppose that $f:\mcX^n \to \mcY$ represents a summary of $\mcX^n$, and let $\mcF$ be a $\sigma$-algebra of subsets of $\mcY$ so that the pair $(\mcY, \mcF)$ is a measurable space.  From a probabilistic perspective, a {\it privacy mechanism} is a family of probability measures $\{ \mu_X: X \in \mcX^n\}$ over $\mcY$.  We can now define what we mean when we say the mechanism satisfies $\epsilon$-DP. While DP was originally introduced in \citet{Dwork2006:Sensitivity}, Definition \ref{DP} is similar to the versions given in \citet{Wasserman2010:StatisticalFDP} and \citet{Kifer2010}.

\begin{defn}[Differential Privacy: \citealp{Dwork2006:Sensitivity}]\label{DP}
A privacy mechanism $\{\mu_X : X \in \mcX^n \}$  
satisfies $\ep$-Differential Privacy ($\ep$-DP) if for all $B \in \mcF$ and  adjacent $X,X' \in \mcX^n$,  
\[
\mu_X(B) \leq \mu_{X'}(B) \exp(\ep).
\]
\end{defn}


From Definition \ref{DP}, we see that, for an $\epsilon$-DP mechanism, $\mu_X$ and $\mu_{X'}$ must be {\it equivalent measures} (i.e. they agree on sets of measure zero) if $\de(X,X')=1$. By transitivity, it follows that $\mu_X$ and $\mu_Y$ are equivalent measures for any $X,Y\in \mathcal X^n$. By the Radon-Nikodym Theorem, we can always therefore interpret DP in terms of densities with respect to a common base measure, $\nu$ (if needed, one can always take $\nu = \mu_X$ for an arbitrary $X\in \mathcal X^n$). 

\begin{prop}
  \label{DPDensity}
Let  $\mscr M=\{\mu_X\mid X\in \mathcal X^n\}$ be a privacy mechanism over a measurable space $(\mathcal Y, \mscr F)$.  Then $\mscr M$ achieves $\epsilon$-DP if and only if there exists a base measure $\nu$ such that $\mu_X\ll \nu$ for all $X\in \mathcal X^n$ and the densities $\{f_X: X \in \mcX^n\}$ (Radon-Nikodym derivatives) of the $\mu_X$ (with respect to $\nu$) satisfy
\[
 f_X(b) \leq f_{X'}(b) \exp(\ep), \]
$\nu$-almost everywhere and for all adjacent $X,X' \in \mcX^n$. 
\end{prop}
\begin{proof}
The reverse direction is given in Remark 1 from \citet{Hall2013}, though we provide the argument here again for completeness. Let $B\in \mscr F$ and $X,X'\in \mathcal X^n$ be adjacent elements. Then 
\begin{align*}
\mu_X(B)&= \int_B f_X(b) \ d\nu(b)
=\int_B \frac{f_{X'}(b)}{f_{X'(b)}} f_{X}(b) \ d \nu(b)\leq \int_B \exp(\ep) f_{X'}(b)\ d\nu(b)
=\exp(\ep) \mu_{X'}(b).
\end{align*}

Going in the other direction we will use a proof by contradiction.  Assume that $\mathscr M$ is an $\epsilon$-DP mechanism.  Recall that two measures are equivalent if they agree on the zero sets, thus, as we have said, the measures in a DP mechanism must all be equivalent.  So, we can assume that all of the measures have a density with respect to some common base measure, $\nu$, which, without loss of generality, we can take to be one of the elements of $\mathscr M$.  
Now assume that there exists a set $B$ and some adjacent databases $X,X'$ such that $f_X(b) > f_{X'}(b) \exp(\epsilon)$ for all $b \in B$ and that $\nu(B) > 0$.  Then this would imply the strict inequality
\begin{align*}
\mu_X(B) &= \int_B f_X(b) \ d \nu(b)> \exp(\epsilon)\int_B \  f_{X'}(b) \ d \nu(b) = \exp(\epsilon)\mu_{X'}(B),
\end{align*}
which is a contradiction, and thus the claim holds.
\end{proof}
Interpreting DP in terms of densities is common in the DP literature (e.g. \citealp{Dwork2014:AFD, Kifer2012:PrivateCERM}), however, we could not find a reference for the precise statement and proof, especially for the reverse implication. 

\section{Exponential Mechanism}\label{s:exponential}
One of the earliest mechanisms designed to satisfy $\ep$-DP, is the \emph{exponential mechanism}, introduced by \citet{McSherry2007}. 
It uses an objective or utility function, which in practice, can be the same objective function used for a (non-private) statistical or machine learning analysis, thus making it especially easy to link DP with existing inferential tools. A simple proof for Proposition \ref{ExponentialMechanism} can be found in \citet{McSherry2007}.
\begin{prop}
  [Exponential Mechanism: \citealp{McSherry2007}]\label{ExponentialMechanism}
Let $(\mathcal Y, \mscr F, \nu)$ be a measure space. Let $\{ \xi_X: \mathcal Y \rightarrow \RR \mid X\in \mathcal X^n\}$ be a collection of measurable functions.  We say that this collection has a 
 finite \emph{sensitivity} $\Delta_{\xi}$, if 
\[ 
|\xi_X(b) - \xi_{X'}(b)| \leq \Delta_{\xi} < \infty,\]
for all adjacent $X,X'$ and $\nu$-almost all $b$. 
If $\int_{\mathcal Y} \exp(\xi_X(b)) \ d\nu(b)<\infty$ for all $X\in \mathcal X^n$, then the collection of probability measures $\{\mu_X\mid X\in \mathcal X^n\}$ with densities $f_X$ (with respect to $\nu$) satisfying 
\[f_X(b) \propto \exp\l[\l(\frac{\ep}{2\Delta_{\xi}}\r) \xi_X(b)\r]\]
satisfies $\ep$-DP.
\end{prop}

We call the set $\{\xi_X\mid X\in \mathcal X^n\}$ the \emph{Objective Function}, used in the exponential mechanism. Note that in Proposition \ref{ExponentialMechanism}, if $\nu$ is a finite measure, $\De(\xi)<\infty$, and $\xi_X(b)$ is bounded above for all $X\in \mathcal X'$ and $\nu$-almost all $b$, then one immediately has $\int \exp(\xi_X(b)) \ d\nu(b)<\infty$.  We will exploit this fact later on as our base measures in infinite dimensions will actually be taken from Gaussian processes, not from any form of Lebesgue measure.

The exponential mechanism offers a general approach to building DP mechanisms, and in fact, any DP mechanism can be expressed as an instantiation of the exponential mechanism, by taking the objective function to be the log-density of the mechanism \citep{McSherry2007}. We remark that the factor of 2 in the exponential mechanism can sometimes be removed (e.g. location families).

Since the solution to many statistical problems can be expressed as the optimizers of some expression, it is natural to set the objective function in the exponential mechanism to this expression. Often times, such expressions can be expressed as empirical risks, such as 
the MLE/MAP estimate \citep{Wang2015:PrivacyFree}, principal component analysis \citep{Chaudhuri2013}, and quantiles of one-dimensional statistics \citep{Smith2011:Privacy-preservingSE}. The following result shows that for objective functions of such forms, the noise added by the exponential mechanism is asymptotically normal.

\begin{thm}[Utility of Exp Mech]\label{thm:utility1}
Assume the observed record, $X_1,\dots,X_n$, and corresponding sequence of objective functions $\xi_n(b):=\xi_X(b)$, for $b \in \mathbb R^d$ satisfy
\begin{enumerate}
    \item $-n^{-1}\xi_n(b)$ 
    are twice differentiable convex functions and there exists 
    a finite $\alpha>0$ 
    such that the eigenvalues of $-n^{-1} \xi_n(b)''$ are greater than $\alpha$ for all $n$ and $b\in \mathbb R^d$;
    \item the minimizers satisfy $\hat b \to b^\star \in \mathbb R^d$ and $-n^{-1} \xi_n(\hat b)'' \to \Sigma^{-1}$ where $\Sigma$ is a $p \times p$ positive definite matrix;
    \item $\xi_n$ has finite sensitivity $\De$, which is constant in $n$.
\end{enumerate}
Assume the base measure has a 
bounded, differentiable density $g(b)$ which is strictly positive in a neighborhood of $b^\star$. 
Then the sanitized value $\tilde b$ drawn from the exponential mechanism with privacy parameter $\ep$ is asymptotically normal
\[
\sqrt{n}(\tilde b -\hat b) \overset{D}{\to} N_p\l(0, \l(\frac {2\De} \ep\r)\Sigma\r).
\]
\end{thm}
\begin{proof}
The density of the exponential mechanism can be expressed as
\[
f_X(b) = c_n^{-1} \exp \left \{\frac \ep {2\De} \xi_X(b)\right \} g(b),
\]
where $c_n$ is the normalizing constant.  
Define the random variable $Z = \sqrt{n}(\tilde b - \hat b)$, then its density is given by
\[
f_n(z) = c_n^{-1} n^{-1/2} g(\hat b + z/\sqrt{n}) \exp \left \{ \frac{\ep}{2\De}\xi_n(\hat b + z/\sqrt{n})\right \}.
\]
We now aim to show that, for $z$ fixed, the density converges to a multivariate normal.  
Using a two term Taylor expansion, we have by Assumption (2) and (3) that
\begin{align*}
 \xi_X(\hat b + z/\sqrt{n})=   [ \xi_X(\hat b) + z^\top \xi_X'(\hat b) / \sqrt{n}+ z^\top \xi_X''(\hat b) z / 2n] + o(1).
\end{align*}
The first term will be absorbed into the constants, since it does not depend on $z$, while the second term is zero for $n$ large, leaving only the third term to contribute to the form of the density.  
Obviously $|g(\hat b + z/\sqrt{n}) - g(b^\star)| \to 0$, so the only remaining task is to show that the combined constants behave appropriately.  Recall that
\begin{align*}
 &c_n n^{1/2} \exp\left\{- \frac{\epsilon}{2\Delta}\xi_n(\hat b)\right\}  =  \int_{B_n} g(\hat b + z/\sqrt{n}) \exp \left \{ \frac{\ep}{2\De}[\xi_n(\hat b + z/\sqrt{n}) - \xi_n(\hat b)]\right \} \ dz.
\end{align*}
By Assumption (1) we have that
\[
\xi_X(\hat b +z/\sqrt{n})- \xi_n(\hat b)
\leq   - \frac{\alpha}{2 } \| z\|^2.
\]
Since $\exp\{-\|z\|^2\}$ is integrable, we can apply the dominated convergence theorem to conclude that the constants converge to something nonzero as well.

Putting everything together, we can conclude that
\begin{align*}
f_n(z) \to f(z) \propto \exp\left \{-\frac\ep{2\De}
z^\top \Sigma^{-1} z/2 
\right\}.
\end{align*}
which, is the density of the multivariate normal.  Applying Scheffe's Theorem, we thus have both convergence in distribution as well as convergence in total variation:
\[
\sqrt{n}(\tilde b - \hat b)
\overset{D}{\to} N_p\l(0, \frac\ep{2\De}\Sigma\r).  
\qedhere\]
\end{proof}

The previous result shows that under common conditions, the noise added by the exponential mechanism is of order $O(1/\sqrt n)$. We know by the theory of M-estimators that the non-private solution to the objective functions $\hat b$ also converges at rate $O(1/\sqrt n)$. So, we have that the use of the exponential mechanisim in such cases preserves the $1/\sqrt n$ convergence rate, but with a sub-optimal asymptotic variance. This means that asymptotically, to achieve the same performance as the non-private estimator, the exponential mechanism requires $k$ times as many samples, where $k$ is some constant larger than $1$, which depends on $\ep$ and $\De$. However, we know that for many problems, it is possible to construct DP mechanisms which only introduce $O(1/n)$ noise, thus having equivalent asymptotics to the non-private estimator
(e.g. \citealp{Smith2011:Privacy-preservingSE,Awan2018:Binomial}). Even though in these settings, the noise is asymptotically negligible, developing accurate approximations is still a challenge, which \citet{Wang2018} recently tackled.

In the next result, we extend Theorem \ref{thm:utility1} from $\RR^p$ to Hilbert spaces. However, we currently only consider base measures which are Gaussian processes.

\begin{thm}[Utility of Exp Mech]\label{thm:utility2}
Suppose that the observed record, $X_1,\dots,X_n$, and objective function $\xi_X(b)$, for $b \in \mathcal H$ satisfy
\begin{enumerate}
    \item $-n^{-1}\xi_n(b)$ 
    are twice differentiable convex functions and there exists a finite 
    $\alpha>0$ 
    such that the eigenvalues of $-n^{-1} \xi_n(b)''$ are greater than $\alpha$ for all $n$ and $b\in \mathcal H$;
    \item the minimizers satisfy $\hat b \to b^\star \in \mathcal H$ and $-n^{-1} \xi_n(\hat b)''^{-1} \to \Sigma$ where $\Sigma$ is positive definite trace class operator (and convergence is wrt this space);
    \item $\xi_n$ has finite sensitivity $\De$, which is constant in $n$.
\end{enumerate}
Assume the base measure is taken to be a Gaussian process, $\nu \sim N_{\mathcal H}(0,C)$, such that $C^{-1}\Sigma$ is bounded. 
Then the sanitized estimate $\tilde b$ is asymptotically normal
\[
\sqrt{n}\l(\tilde b -\hat b\r) \overset{D}{\to} N_{\mathcal H}\l(0, \frac{2\Delta}{\ep}\Sigma\r).
\]
\end{thm}
\begin{proof}
The proof would be essentially the same as before, however when changing variables via standardizing, the base measure is no longer Lebesgue and thus the effects of rescaling the base measure cannot be ignored.  Recall there is no translation invariant $\sigma$-finite measure in infinite dimensions.  Consider $Z = \sqrt{n}(\tilde b- \hat b)$ and
\begin{align*}
P(\sqrt{n}(\tilde b- \hat b) \in A) = \int_{ \hat b + A/\sqrt{n}} f_X(b) \ d \nu(b)= n^{-1/2} \int_A f_X(\hat b + z/\sqrt{n}) \ d \nu(\hat b + z/\sqrt{n}).
\end{align*}
The same Taylor expansion arguments from before still apply, however the base measure has now been shifted and scaled.  In particular, if $d \tilde \nu (z) = d \nu(\hat b + z/\sqrt{n})$, then $\tilde \nu$ is the measure of a Gaussian process with mean $-\sqrt{n} \hat b$ and covariance operator $n C$.  So, we have that, up to a normalizing constant
\begin{align*}
P(\sqrt{n}(\tilde b- \hat b) \in A)
&\approx \int_A \exp\l\{ \l\langle z, \frac{\ep}{2\Delta}\xi_X''(\hat b) z \r\rangle /2n \r\} \ d \tilde  \nu(z)\approx \int_A \exp\l\{ -\l\langle z, \frac{\ep}{2\Delta}\Sigma^{-1} z \r\rangle /2 \r\} \ d \tilde  \nu(z).
\end{align*}
However, this is a Gaussian measure with covariance operator $(\frac{\ep}{2\Delta}\Sigma^{-1} + C^{-1}/n)^{-1} $ and mean $- n^{-1/2} (\frac{\ep}{2\Delta}\Sigma^{-1}  + C^{-1}/n)^{-1}C^{-1} \hat b $.  Since $C$ is fixed, the following limits hold
\begin{align*}
\l(\frac{\ep}{2\Delta}\Sigma^{-1} + C^{-1}/n\r)^{-1}  \to \frac{2\Delta}{\ep}\Sigma
\end{align*}
\begin{align*}
- n^{-1/2} \l(\frac{\ep}{2\Delta}\Sigma^{-1} + C^{-1}/n\r)^{-1}C^{-1} \hat b \to 0.
\end{align*}\qedhere
\end{proof}

\begin{remark}
 The requirement that $C^{-1}\Sigma$ is bounded implies that the base measure is ``rougher'' than the asymptotic distribution of $\hat b$.  One way to view the assumption on  $\xi_X''$ is through tightness.  In particular, if one assumed only that $\Sigma$ was bounded, then the sequence of measures need not be tight and thus one does not get convergence in the ``strong topology" in $\mathcal H$ \citep[Remark 3.3]{billingsley2013convergence,chen1998central}.  However, one could still obtain convergence of properly normalized continuous linear functionals. 
\end{remark}

\begin{example}
Consider $X_1,\dots, X_n \in \mathcal H$ are drawn from a Gaussian process with mean $\mu_X$ and covariance operator $C_X$.  Consider estimating $\mu_X$ using the target function
\[
-\xi_X(b) = \sum_{i=1}^n \| X_i - b\|^2.
\]
Assume that the $\|X_i\| \leq 1$ and thus we need only consider $\| b \| \leq 1$.  In that case, the sensitivity is bounded by 4.  However, for this target function the exponential mechanism will not be asymptotically Gaussian (in the strong topology).  If we consider the second derivative we have
\[
-\xi_X''(b) = 2 n I,
\]
and thus $(-\xi_X''(b)/n)^{-1} = (1/2) I$, which is not a nuclear operator in infnite dimensions.  However, if instead we consider the penalized version
\[
-\xi_X(b) = \sum_{i=1}^n \| X_i - b\|^2 + n \lambda \| b\|^2_{C},
\]
where $\| b\|^2_C = \langle b, C^{-1} b\rangle$, then the sensitivity is the same, but the second derivative is now
\[
-\xi_X''(b) = 2 n I + 2 n \la C^{-1},
\]
which satisfies the assumptions of Theorem \ref{thm:utility2}.  In this case, we can now take our Guassian mechanism to be a mean zero Gaussian process with covariance $C$.
\end{example}

We stress that, in finite samples, there is no issue related to privacy even when $\xi_X''(b)^{-1}$ is not nuclear since we are assuming the mechanism is defined using a value probability distribution as the base measure.  What the previous results and this example illustrate is that there is a price to pay for using such a flexible mechanism.  In the ``good" case, when the assumptions of Theorem \ref{thm:utility2} are met, one has an asymptotically non-negligle noise, but in the "bad" case, the noise can be even larger, since the covariance operator can blow up. 

\section{DP Functional Principal Components}
\label{s:pca}
In this section, we apply the exponential mechanism to the problem of producing private functional principal component analysis (FPCA).

Let $(\mathcal H, \langle\cdot, \cdot \rangle)$ be a Hilbert Space. Let $X \in \mathcal H^n$ be such that its components satisfy $\lVert X_i \rVert \leq 1$ for all $i=1,\ldots, n$. Call $\hat S(X)$ the $k$-dimensional subspace of $\mathcal H$ given by the span of the first $k$ principal components of $X$. Let $P_{\hat S(X)}: \mathcal H \rightarrow \mathcal H$ denote the projection operator of $\mathcal H$ onto $\hat S(X)$. We can write $P_{\hat S(X)}$ as the solution to the optimization problem 
\begin{equation}\label{eq:pca}
P_{\hat S(X)} = \arg\min_{P \in \mscr P_k} \sum_{i=1}^n \lVert X_i - PX_i \rVert^2,
\end{equation}
where $\mscr P_k$ is the set of projection operators $P: \mathcal H \rightarrow \mathcal H$ of rank $k$. Equivalently, we can write 
\[P_{\hat S(X)} = \arg\max_{P \in \mscr P_k} \sum_{i=1}^n \lVert PX_i \rVert^2.\]
More specifically, in this section, we develop a set of probability measures $\mscr M$ on $\mscr P_k$, indexed by $\mathcal H^n$, which satisfy $\ep$-DP, and such that a random element $P$ from $\mu_X \in \mscr M$ is ``close'' to $P_{\hat S(X)}$.

Our approach follows that of \citet{Chaudhuri2013}. We take our objective function to be $\xi:\mathcal X^n \times \mscr P_k \rightarrow \RR$ defined by $\xi_X(P) = \sum_{i=1}^n \lVert PX_i \rVert^2$. Note that $\Delta_{\xi}=1$, since $\lVert PX_i \rVert^2\leq\lVert X_i \rVert^2\leq 1$ for any $P\in \mscr P_k$ and any $i=1,\ldots, n$. Since $\sum_{i=1}^n \lVert PX_i \rVert^2\leq n$, for any probability measure $\nu$ on $\mscr P_k$, the class of densities on $\mscr P_k$ with respect to $\nu$ given by
\[f_X(P) \propto \exp\l(\frac{\ep}{2} \sum_{i=1}^n \lVert PX_i \rVert ^2\r), \text{ satisfies $\ep$-DP.}\]

If $\mathcal H$ is finite dimensional, then $\mscr P_k$ is a compact subset of the space of linear operators (e.g. matrices when $\mathcal H = \mathbb R^p$). In that case, there exists a uniform distribution on $\mscr P_k$. In \citet{Chaudhuri2013}, they implement the exponential mechanism as above, with respect to the uniform distribution on $\mscr P_k$.

For arbitrary $\mathcal H$, $\mscr P_k$ is not compact, so we must find another base measure on $\mscr P_k$. To understand our proposed construction, we again consider the finite dimensional $\mathcal H$. Let $P \sim \mathrm{Unif}(\mscr P_k)$, that is $P$ is drawn from the uniform distribution on $\mscr P_k$. Let $V_1,\ldots, V_k \iid N(0, I)$, be iid multivariate normal with mean zero and identity covariance matrix. Then $P \overset d = \mathrm{Projection}(\vspan(V_1,\ldots, V_k))$ (since $V_k$ is invariant under rotations). From this factorization, a natural extension for arbitrary $\mathcal H$ becomes clear. Let $V_1,\ldots, V_k \iid N_{\mathcal H}(0, \Sa)$, be iid Gaussian processes in $\mathcal H$ with zero mean and covariance operator $\Sa$.  Note that $\Sa$ must be positive definite and  trace class, which excludes the identity when $\mathcal H$ is infinite dimensional 
\citep{bogachev1998gaussian}.  We can also tailor $\Sigma$ to instill certain properties such as smoothness or periodicity. Then set $P = \mathrm{Projection}(\vspan(V_1,\ldots, V_k))$. This procedure induces in a probability measure on $\mscr P_k$, which we call $\nu$.

\begin{thm}
\label{ExpMechHilbert_v2}
	Let $\mathcal H$ be a real separable Hilbert Space and $\mathscr P_k$ the collection of all $k$-dimensional projection operators over $\mathcal H$.  Let $\nu$ be the probability measure over $\mscr P_k$ induced by the transformation Projection(span($V_1,\dots,V_k$)), where $V_i \in \mathcal H$ are iid Gaussian process with mean 0 and covariance operator $\Sa$.  Let $\mscr M$ be the class of probability measures on $\mscr P_k$ with densities
	\[f_X(P) \propto \exp\l(\frac{\ep}{2} \sum_{i=1}^n \lVert P X_i\rVert^2\r)\]
	with respect to $\nu$. Then $\mscr M$ satisfies $\ep$-DP.
\end{thm}

\begin{thm}\label{ExpMechHilbert}
  Let $\mathcal H$ be a Hilbert Space,  $k<n$ be two positive integers, and $\mscr M$ be the class of probability measures on $\mathcal H^k$ with densities $f_X(V_1,\ldots, V_k)$ proportional to 
\[  \exp\l(\frac{\ep}{2} \sum_{i=1}^n \lVert \mathrm{Projection}(\vspan(V_1,\ldots, V_k))X_i\rVert^2\r)\]
with respect to $\nu$ (the measure induced by the Gaussian distribution $N^k(0, \Sa)$) on $\mathcal H^k$. 
Then $\mscr M$ satisfies $\ep$-DP. 
\end{thm}

As we see in the next section, we will represent the output of Theorem \ref{ExpMechHilbert} as an arbitrary basis for a $k$-dimensional subspace $\twid S$ of $\mathcal H$, which can then be assembled into a projection operator as needed.

\begin{remark}
  In many cases we can interpret $\Sigma$ as instilling some particular structure on $P$ or the $V_i$.  For example, if  $\mcH = L^2[0,1]$, then we could define $\Sa$ using the kernel of an RKHS. The kernel could then be chosen so that $\twid S$ have a certain number of derivatives as many Sobelev spaces are RKHS as well \citep{berlinet2011reproducing}, which is often a natural assumption.
\end{remark}

\section{PCA continued: Sampling}\label{s:sampling}
In the previous section, we developed a set of $\ep$-DP probability measures for arbitrary Hilbert Spaces. However, for these to be of use to us, we need to be able to sample from these probability distributions. As is common in FDA \citep{ramsay2005springer,kokoszka2017introduction} we use finite dimensional approximations via basis expansions for computation.

Let $b_1,b_2,\ldots$ be an orthonormal basis for $\mathcal H$. We will work in the $m$-dimensional subspace $\mathcal H_m = \vspan(b_1,\ldots, b_m)$. Given our observed values $X_i \in \mathcal H$, call $X_{ij} = \langle X_i, b_j\rangle$ for $i=1,\ldots, n$ and $j=1,\ldots,m$. We arrange these real values in an $n \times m$ matrix ${\bf X} = (X_{ij})$.
Next, let $\Sa$ be a trace class covariance operator on $\mathcal H$. Write $\Sa_{ij} = \langle b_i, \Sa b_j\rangle$ for $i,j=1,\ldots, m$. We put these values in an $m \times m$ matrix ${\bf\Sa} = (\Sa_{ij})$,
which is a positive definite matrix in $\RR^{m\times m}$. In this setup, we then draw $(V_1,\ldots, V_k)\in \mathcal H_m$. Call $V_{ij} =\langle V_i,b_j\rangle$, and arrange these values into a real-valued matrix ${\bf V} = (V_{ij})$
We then draw from the density $f(V)$, with respect to Lebesgue measure on $\RR^{k\times m}$, which is proportional to 
\begin{align*}
\exp\l(\frac{\ep}{2} \tr(X^\top X V(V^\top V)^{-1} V^\top - V^\top \Sa^{-1} V))\r).
\end{align*}
In fact, we can obtain a more convenient form for sampling. Since we only need the span of $V$, we can condition on the columns of $V$ being orthonormal. The density $f(V\mid \text{orthonormal})$, with respect to the uniform measure on the set of orthonormal matrices in $\RR^{m\times k}$, is proportional to 
\[ \exp\l( \frac{\ep}{2} \tr\l( V^\top \l(X^\top X - \Sa^{-1}\r) V\r) \r),\]
which is an instance of the Matrix-Bingham-Von-Mises distribution, for which an efficient Gibbs sampler is known \citep{hoff2009simulation,hoff2018package}. 
\section{Numerical Studies}\label{s:numerical}
In this section we assess the performance of the exponential mechanism for private FPCA, as developed in Sections \ref{s:pca} and \ref{s:sampling} on both simulated and real data.
\subsection{Simulation Study}
\begin{figure}
    \centering
	    \includegraphics[scale=0.23]{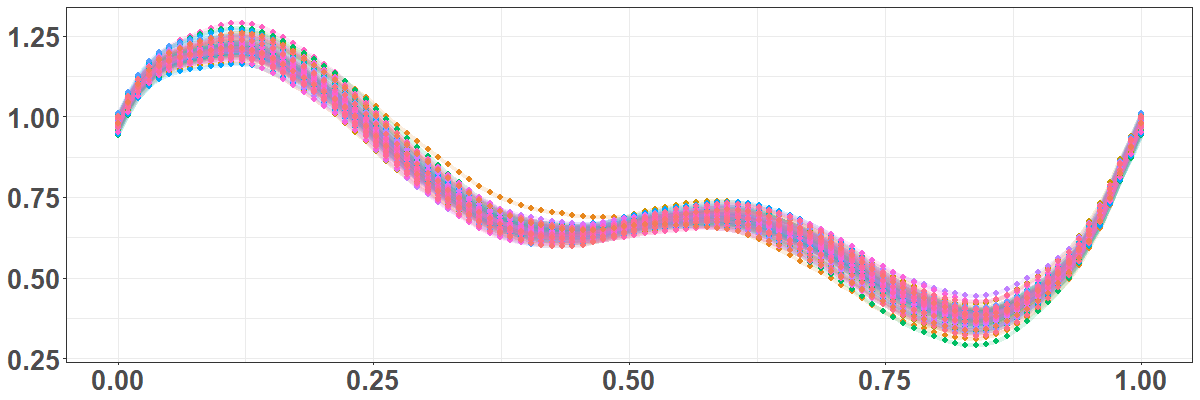}
\caption{Plot of 100 curves generated for the  simulation.}\label{FourierCurves}
\end{figure}
For our simulation study, we generated data on a grid of $100$ evenly spaced points on $[0,1]$ using the Karhunen-Loeve expansion with Gaussian noise added:
\[X_i(t_{ik})=\mu(t_{ik})+\sum_{j=1}^p\frac{1}{j^2}U_{ij}u_j(t_{ik}) + \varepsilon_{ik}, 
\]
for $i=1,\dots,n$, $k=1,\dots,100$.  The $u_j(t)$ are the true functional principal components, $\varepsilon_{ik}$ are independent errors sampled from the Gaussian distribution $N(0,1)$, and scores $U_{ij}$ are sampled from $N(0,0.1)$. Note that for each scenario we re-scale the $X_i$ so that $||X_i||^2<1$ for $i=1,...,n$. 

The $u_j(t)$ are comprised of Fourier basis functions and to fully explore the effectiveness of this approach, we vary the sample size $n$, privacy budget $\epsilon$, and repeat each scenario 10 times. Data is generated using $p=21$ true components and additional weights were placed on the fourth term in the Fourier expansion, creating the overall shape shown in Figure~\ref{FourierCurves}. We release only $k=1$ components.

\begin{figure}
	\centering
        \begin{subfigure}[b]{0.45\textwidth}
              \includegraphics[scale=0.2]{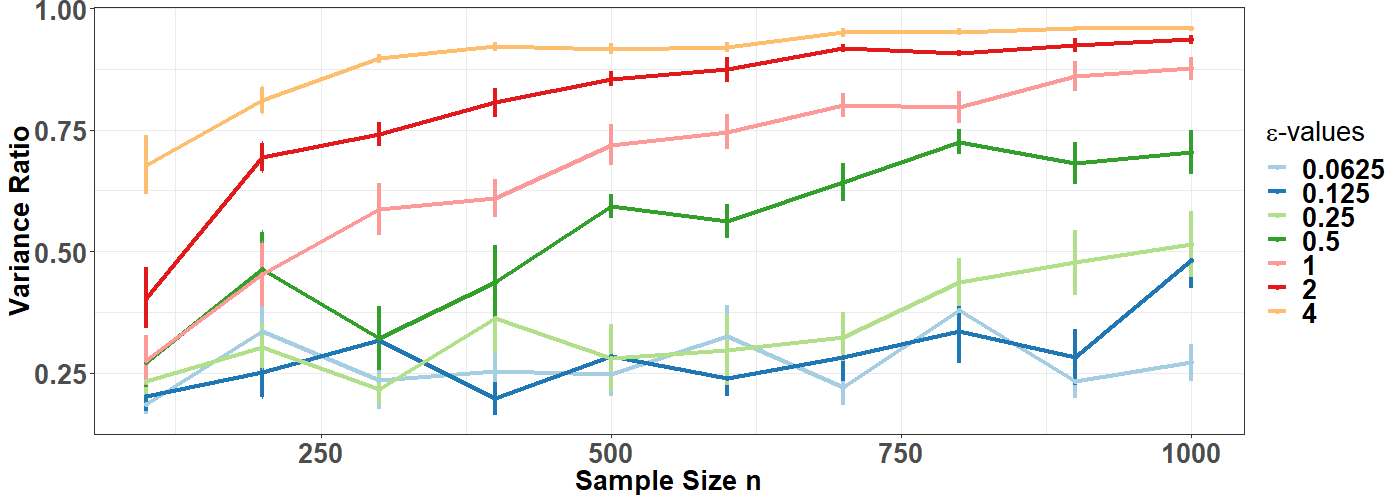}
              \caption{
              Average ratio of variance explained between the private and non-private principal components.
              }\label{VarianceRatioPeriodic}
          \end{subfigure}
          \begin{subfigure}[b]{0.45\textwidth}
              \includegraphics[scale=0.2]{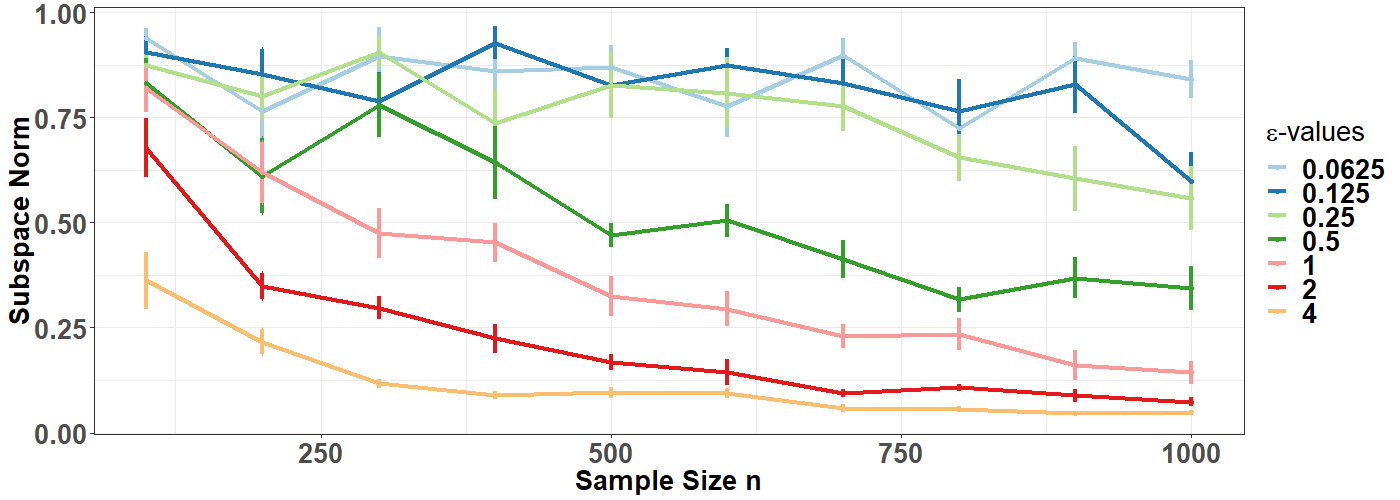}
              \caption{
              Average subspace norm of private principal components. 
              }\label{SubspaceNormPeriodic}
          \end{subfigure}
          \caption{Average performance measurements in simulation scenarios over sample sizes ranging from $n=100$ to $1000$. Standard error bars are provided at each point. 
          }
          \label{summarymeasuresperiodic}
\end{figure}
\begin{figure}
      \centering
              \begin{subfigure}[b]{0.45\textwidth}
              \includegraphics[scale=0.2]{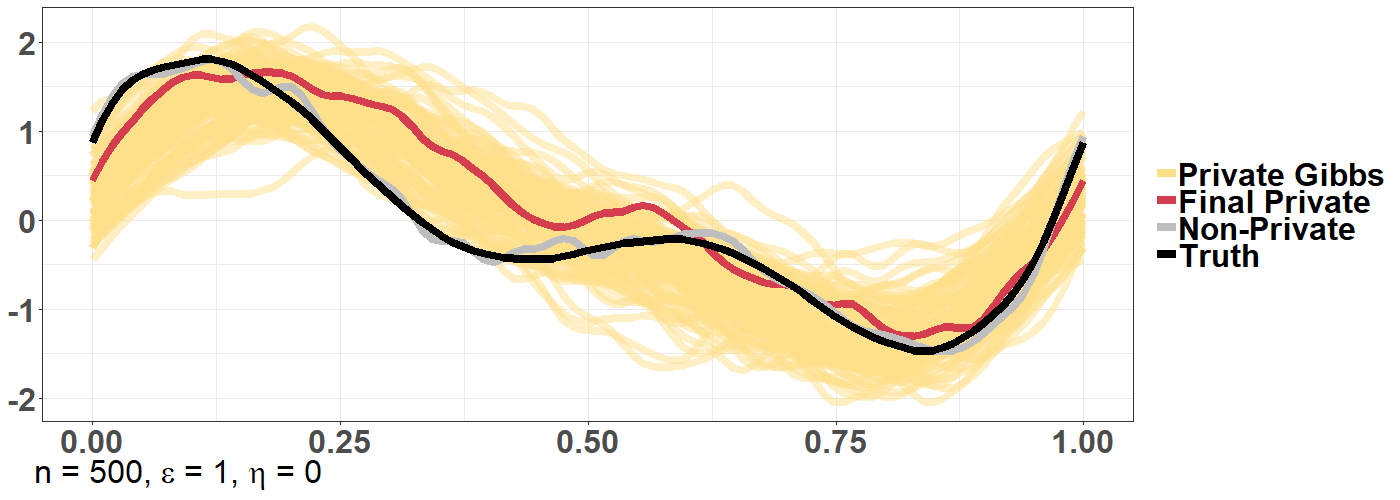}
              \caption{
              One instance of the simulation where $n=500$ and $\epsilon=1$. 
              }\label{n500e1et0periodic}
          \end{subfigure}
          \begin{subfigure}[b]{0.45\textwidth}
              \includegraphics[scale=0.2]{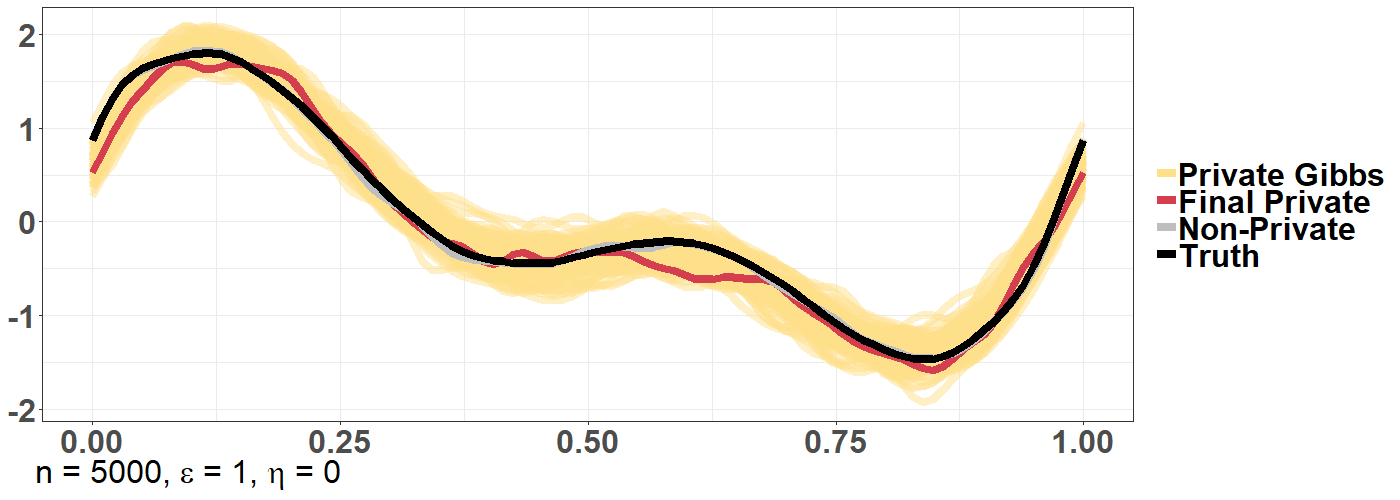}
              \caption{One instance of the simulation where $n=5000$ and $\epsilon=1$}\label{n5000e1et0periodic}
          \end{subfigure}
\caption{Comparisons between the private estimate, non-private estimate, and true first functional principal component. 
The last 100 Gibbs updates for the private estimate are provided to demonstrate the variability of the mechanism.
}
\label{summarycurvesperiodic}
\end{figure}

We also specify $m$, the number of orthonormal basis functions $b_i$, when restricting the functional observations to a finite dimensional space and $\Sigma$, a trace class covariance operator on $\mathcal H$. It is common to take $m$ to be some sufficiently large value, usually around 40-50, and for our simulation scenarios we have $m=40$. The choice of $\Sigma$ can vary depending on what structures one may want to induce on the functions (i.e. the number of derivatives). For our choice of $\Sigma$, we take it to be a diagonal matrix with $\Sigma_{ii}=i^{-3}$, which corresponds to requiring that the $V_i$ are continuous. 
Given that the data is periodic, we chose to use the Fourier basis functions as $b_i$. Finally, recall there is an efficient Gibbs sampler for this approach \cite{hoff2009simulation}, which has been implemented in the \texttt{rstiefel} package \cite{hoff2018package} in R.  This also requires a fixed number of iterations as burn-in prior to starting the procedure. Following the computational experiments in \cite{Chaudhuri2013}, we used 20,000 iterations and had similar convergence results.

We provide two measurements of performance to compare the resulting space of orthogonal projection operators. The first compares the ratio of variability accounted for between the private and non-private estimates of the $k$ functional principal components. More explicitly,
\[0 \leq \frac{||X^T\tilde{P}X||_{F}^2}{||X^T\hat{P}X||^2_F} \leq 1,
\]
where $||\cdot||_F$ is the Frobenius norm, $\tilde{P}$ is the projection onto the span of $V$ drawn from the mechanism in Theorem \ref{ExpMechHilbert}, and $\hat{P}$ the the non-private solution to \eqref{eq:pca}. 

The second measure gives an indication of how close the range of $\tilde{P}$ is to $\hat{P}$:
\[0 \leq \frac12{||\tilde{P}-\hat{P}||_F^2}\leq k.
\]
If the range of $\tilde{P}$ and $\hat{P}$ agree in $h$ dimensions and are orthogonal in $k-h$ dimensions, then this measure gives the value $k-h$. So this can be interpreted as roughly the number of dimensions that $\tilde{P}$ and $\hat{P}$ disagree. 

We summarize the results in Figures~\ref{VarianceRatioPeriodic} and~\ref{SubspaceNormPeriodic} over a range of sample size $n$ and privacy budget $\epsilon$. Note that, as expected, larger sample sizes can preserve utility (in terms of the two measurements described previously) for stricter privacy requirements. Additionally, we provide the sanitized curve for the first principal component for one instance of a scenario with a sample size of $n=500$, and $n=5000$, seen in Figures~\ref{n500e1et0periodic} and~\ref{n5000e1et0periodic}. The last 100 Gibbs updates are given as well, demonstrating the variability in each sample size. Note that even with a privacy budget of $\epsilon=1$ and relatively low sample size, the overall shape is captured, but the variance is reduced when $n=5000$.

\subsection{Applications}
\begin{table}
  \centering
\caption{Average performance for private principal components from the Berkeley growth and DTI data sets. Standard errors are provided in parenthesis for reference.}
\label{applications}
\begin{tabular}{rlll}  
\toprule
& \multicolumn{3}{c}{No. of Components ($k$)} \\
\cmidrule(r){2-4}
& \multicolumn{1}{c}{1} 
& \multicolumn{1}{c}{2}
& \multicolumn{1}{c}{3}\\
\midrule

\textbf{Berkeley} & \multicolumn{3}{c}{Variance Ratio} \\
\cdashlinelr{2-4}
1/8 & 0.264  (.024)	& 0.494  (.023)	& 0.672  (.020) \\
1/4 & 0.343  (.024)	& 0.523  (.023)	& 0.681  (.020) \\
1/2 & 0.408  (.025)	& 0.523  (.022)	& 0.729  (.019) \\
1 & 0.550  (.025)	& 0.680  (.018)	& 0.775  (.015) \\
2 & 0.743  (.018)	& 0.787  (.012)	& 0.855  (.010) \\

\textbf{DTI (\textit{cca})} & \multicolumn{3}{c}{Variance Ratio} \\
\cdashlinelr{2-4}
1/8 & 0.372 (.025)	& 0.569 (.024)	& 0.727 (.018) \\
1/4 & 0.497 (.026)	& 0.676 (.021)	& 0.811 (.011) \\
1/2 & 0.726 (.020)	& 0.812 (.014)	& 0.876 (.009) \\
1 & 0.879 (.009)	& 0.885 (.007)	& 0.910 (.005) \\
2 & 0.933 (.006)	& 0.928 (.004)	& 0.939 (.003) \\
\end{tabular}
\end{table}
\begin{table}
  \centering
\caption{Average performance for private principal components from the Berkeley growth and DTI data sets. Standard errors are provided in parenthesis for reference.}
\label{applications2}
\begin{tabular}{rlll}  
\toprule
& \multicolumn{3}{c}{No. of Components ($k$)} \\
\cmidrule(r){2-4}
& \multicolumn{1}{c}{1} 
& \multicolumn{1}{c}{2}
& \multicolumn{1}{c}{3}\\
\midrule
\textbf{Berkeley}& \multicolumn{3}{c}{Subspace Norm} \\
\cdashlinelr{2-4}
1/8 & 0.776 (.025)	& 1.115 (.036)	& 1.100 (.034) \\
1/4 & 0.701 (.025)	& 1.046 (.035)	& 1.135 (.030) \\
1/2 & 0.633 (.027)	& 1.063 (.033)	& 1.066 (.030) \\
1 & 0.484 (.027)	& 0.883 (.031)	& 0.962 (.032) \\
2 & 0.275 (.020)	& 0.770 (.032)	& 0.938 (.035) \\
\textbf{DTI (\textit{cca})}& \multicolumn{3}{c}{Subspace Norm} \\
\cdashlinelr{2-4}
1/8 & 0.679 (.026)	& 1.098 (.035)	& 1.074 (.030) \\
1/4 & 0.544 (.029)	& 0.976 (.027)	& 1.079 (.029) \\
1/2 & 0.296 (.021)	& 0.861 (.027)	& 0.982 (.030) \\
1 & 0.131 (.010)	& 0.770 (.026)	& 0.940 (.035) \\
2 & 0.073 (.006)	& 0.640 (.030)	& 0.758 (.035) \\
\end{tabular}

\end{table}
For the real data application, we applied our procedure to two data sets, 
the Berkeley growth study from the \texttt{fda} package \cite{fdapackage}, and Diffusion Tensor Imaging (DTI) from the \texttt{refund} package \cite{refundpackage}.
The Berkely data has the heights of 93 children at 31 time points with ages ranging from 1-18.
DTI gives fractional anisotropy (FA) tract profiles for the corpus callosum (CCA) the right corticospinal tract (RCST) for patients with Multiple Sclerosis as well as controls. We focus on the \textit{cca} data, which includes 382 patients measured at 93 equally spaced locations along the CCA.

Results are summarized in Tables~\ref{applications} and \ref{applications2} when releasing 1-3 principal components across a range of privacy budgets and averaging the performance measurements over 100 repetitions of our procedure. For each data set we selected the Gaussian kernel for $\Sigma$ with a smoothness parameter that requires $m=5$ eigenvalues to explain $>$99\% of variation. Its corresponding eigenfunctions were selected for the orthonormal basis $b_i$. Our approach is more effective over the DTI data set, which may be due to the true variation explained by the non-private components. For DTI the cumulative variation is .77, .86, and .93 for the top 1, 2, and 3 components respectively, while Berkeley has 0.82, 0.95, and 0.98. When things are too ``simple", necessary deviations for privacy show more loss in variation explained compared to the non-private estimates. Overall, this still demonstrates the effectiveness of our procedure under different types of real data with smaller sample sizes.

\section{Discussion}\label{s:conclusions}
In this paper, we studied the exponential mechanism in the setting of separable Hilbert spaces. We showed that generally when the objective  is an empirical risk function, the exponential mechanism has a CLT implying that asymptotically non-negligible noise is introduced. Since the exponential mechanism is popularly used, this result demands the following question: what properties of the objective function guarantee asymptotically negligible noise?

Through our simulations and applications, we found that the choice of $\Sigma$ can have a significant impact on the result of the private FPCA analysis. In particular, $\Sa$ can be rescaled by any positive constant, which affects the smoothing but does not change the interpretation in terms of number of derivatives. While our approach requires that $\Sa$ is chosen before seeing the data, it would be preferable to have a method of learning $\Sa$ within the DP procedure. Future researchers should investigate effective methods of tuning parameters under DP.

In the data applications, we found that counter-intuitively, our DP FPCA approach performs better when there is more variability in the data. Perhaps this is because our measures of performance are comparing the DP estimates to the non-private estimates, and the variability hurts both. It would be worth while to investigate this further to better understand how variability in the data affects the performance of DP methods.

\section*{Acknowledgments}
This research was supported in part by the following grants to Pennsylvania State University: NSF Grant SES-1534433, NSF Grant DMS-1712826, and NIH Grant 5T32LM012415-03 via the  Biomedical Big Data to Knowledge (B2D2K) Predoctoral Training Program.



\bibliographystyle{icml2019} 
\bibliography{./DataPrivacyBib}{}

\begin{thebibliography}{40}
\providecommand{\natexlab}[1]{#1}
\providecommand{\url}[1]{\texttt{#1}}
\expandafter\ifx\csname urlstyle\endcsname\relax
  \providecommand{\doi}[1]{doi: #1}\else
  \providecommand{\doi}{doi: \begingroup \urlstyle{rm}\Url}\fi

\bibitem[Alda \& Rubinstein(2017)Alda and Rubinstein]{Alda2017}
Alda, F. and Rubinstein, B.~I.
\newblock The bernstein mechanism: Function release under differential privacy.
\newblock In \emph{AAAI}, pp.\  1705--1711, 2017.

\bibitem[Awan \& Slavkovi\'{c}(2018)Awan and Slavkovi\'{c}]{Awan2018:Binomial}
Awan, J. and Slavkovi\'{c}, A.
\newblock Differentially private uniformly most powerful tests for binomial
  data.
\newblock In \emph{Advances in Neural Information Processing Systems 31}, pp.\
  4212--4222. Curran Associates, Inc., 2018.

\bibitem[Awan \& Slavkovi\'c(2018)Awan and Slavkovi\'c]{Awan2018:Structure}
Awan, J. and Slavkovi\'c, A.
\newblock Structure and sensitivity in differential privacy: Comparing $k$-norm
  mechanisms.
\newblock \emph{ArXiv e-prints}, January 2018.
\newblock Submitted.

\bibitem[Berlinet \& Thomas-Agnan(2011)Berlinet and
  Thomas-Agnan]{berlinet2011reproducing}
Berlinet, A. and Thomas-Agnan, C.
\newblock \emph{Reproducing kernel Hilbert spaces in probability and
  statistics}.
\newblock Springer Science \& Business Media, 2011.

\bibitem[Billingsley(2013)]{billingsley2013convergence}
Billingsley, P.
\newblock \emph{Convergence of probability measures}.
\newblock John Wiley \& Sons, 2013.

\bibitem[Blum et~al.(2005)Blum, Dwork, McSherry, and Nissim]{Blum2005}
Blum, A., Dwork, C., McSherry, F., and Nissim, K.
\newblock Practical privacy: the sulq framework.
\newblock In \emph{Proceedings of the twenty-fourth ACM SIGMOD-SIGACT-SIGART
  symposium on Principles of database systems}, pp.\  128--138. ACM, 2005.

\bibitem[Bogachev(1998)]{bogachev1998gaussian}
Bogachev, V.~I.
\newblock \emph{Gaussian measures}.
\newblock Number~62. American Mathematical Soc., 1998.

\bibitem[Canonne et~al.(2018)Canonne, Kamath, McMillan, Smith, and
  Ullman]{Canonne2018}
Canonne, C.~L., Kamath, G., McMillan, A., Smith, A.~D., and Ullman, J.
\newblock The structure of optimal private tests for simple hypotheses.
\newblock \emph{CoRR}, abs/1811.11148, 2018.

\bibitem[Chaudhuri \& Monteleoni(2009)Chaudhuri and
  Monteleoni]{Chaudhuri2009:Logistic}
Chaudhuri, K. and Monteleoni, C.
\newblock Privacy-preserving logistic regression.
\newblock In Koller, D., Schuurmans, D., Bengio, Y., and Bottou, L. (eds.),
  \emph{Advances in Neural Information Processing Systems 21}, pp.\  289--296.
  Curran Associates, Inc., 2009.

\bibitem[Chaudhuri et~al.(2011)Chaudhuri, Monteleoni, and
  Sarwate]{Chaudhuri2011:DPERM}
Chaudhuri, K., Monteleoni, C., and Sarwate, D.
\newblock Differentially private empirical risk minimization.
\newblock In \emph{Journal of Machine Learning Research}, volume~12, pp.\
  1069--1109, 2011.

\bibitem[Chaudhuri et~al.(2013)Chaudhuri, Sarwate, and Sinha]{Chaudhuri2013}
Chaudhuri, K., Sarwate, A.~D., and Sinha, K.
\newblock A near-optimal algorithm for differentially-private principal
  components.
\newblock \emph{Journal of Machine Learning Research}, 14\penalty0
  (1):\penalty0 2905--2943, January 2013.
\newblock ISSN 1532-4435.

\bibitem[Chen \& White(1998)Chen and White]{chen1998central}
Chen, X. and White, H.
\newblock Central limit and functional central limit theorems for
  hilbert-valued dependent heterogeneous arrays with applications.
\newblock \emph{Econometric Theory}, 14\penalty0 (2):\penalty0 260--284, 1998.

\bibitem[Dwork \& Roth(2014)Dwork and Roth]{Dwork2014:AFD}
Dwork, C. and Roth, A.
\newblock The algorithmic foundations of differential privacy.
\newblock \emph{Found. Trends Theor. Comput. Sci.}, 9\penalty0
  (3\&\#8211;4):\penalty0 211--407, August 2014.
\newblock ISSN 1551-305X.
\newblock \doi{10.1561/0400000042}.

\bibitem[Dwork et~al.(2006)Dwork, McSherry, Nissim, and
  Smith]{Dwork2006:Sensitivity}
Dwork, C., McSherry, F., Nissim, K., and Smith, A.
\newblock \emph{Calibrating Noise to Sensitivity in Private Data Analysis},
  pp.\  265--284.
\newblock Springer Berlin Heidelberg, Berlin, Heidelberg, 2006.
\newblock ISBN 978-3-540-32732-5.

\bibitem[Dwork et~al.(2014)Dwork, Talwar, Thakurta, and Zhang]{Dwork2014}
Dwork, C., Talwar, K., Thakurta, A., and Zhang, L.
\newblock Analyze gauss: Optimal bounds for privacy-preserving principal
  component analysis.
\newblock In \emph{Proceedings of the Forty-sixth Annual ACM Symposium on
  Theory of Computing}, STOC '14, pp.\  11--20, New York, NY, USA, 2014. ACM.
\newblock ISBN 978-1-4503-2710-7.
\newblock \doi{10.1145/2591796.2591883}.

\bibitem[Gaboardi et~al.(2016)Gaboardi, Lim, Rogers, and Vadhan]{Gaboardi2016}
Gaboardi, M., Lim, H., Rogers, R., and Vadhan, S.
\newblock Differentially private chi-squared hypothesis testing: Goodness of
  fit and independence testing.
\newblock In Balcan, M.~F. and Weinberger, K.~Q. (eds.), \emph{Proceedings of
  The 33rd International Conference on Machine Learning}, volume~48 of
  \emph{Proceedings of Machine Learning Research}, pp.\  2111--2120, New York,
  New York, USA, 20--22 Jun 2016. PMLR.

\bibitem[Goldsmith et~al.(2018)Goldsmith, Scheipl, Huang, Wrobel, Gellar,
  Harezlak, McLean, Swihart, Xiao, Crainiceanu, and Reiss]{refundpackage}
Goldsmith, J., Scheipl, F., Huang, L., Wrobel, J., Gellar, J., Harezlak, J.,
  McLean, M.~W., Swihart, B., Xiao, L., Crainiceanu, C., and Reiss, P.~T.
\newblock \emph{refund: Regression with Functional Data}, 2018.
\newblock R package version 0.1-17.

\bibitem[Hall et~al.(2013)Hall, Rinaldo, and Wasserman]{Hall2013}
Hall, R., Rinaldo, A., and Wasserman, L.
\newblock Differential privacy for functions and functional data.
\newblock \emph{Journal of Machine Learning Research}, 14\penalty0
  (1):\penalty0 703--727, February 2013.
\newblock ISSN 1532-4435.

\bibitem[Hoff \& Franks(2018)Hoff and Franks]{hoff2018package}
Hoff, P. and Franks, A.
\newblock \emph{rstiefel: Random Orthonormal Matrix Generation and Optimization
  on the Stiefel Manifold}, 2018.
\newblock R package version 0.20.

\bibitem[Hoff(2009)]{hoff2009simulation}
Hoff, P.~D.
\newblock Simulation of the matrix bingham--von mises--fisher distribution,
  with applications to multivariate and relational data.
\newblock \emph{Journal of Computational and Graphical Statistics}, 18\penalty0
  (2):\penalty0 438--456, 2009.

\bibitem[Imtiaz \& Sarwate(2016)Imtiaz and Sarwate]{Imtiaz2016}
Imtiaz, H. and Sarwate, A.~D.
\newblock Symmetric matrix perturbation for differentially-private principal
  component analysis.
\newblock In \emph{Acoustics, Speech and Signal Processing (ICASSP), 2016 IEEE
  International Conference on}, pp.\  2339--2343. IEEE, 2016.

\bibitem[Jiang et~al.(2013)Jiang, Ji, Wang, Mohammed, Cheng, and
  Ohno-Machado]{Jiang2013}
Jiang, X., Ji, Z., Wang, S., Mohammed, N., Cheng, S., and Ohno-Machado, L.
\newblock Differential-private data publishing through component analysis.
\newblock \emph{Transactions on data privacy}, 6\penalty0 (1):\penalty0 19,
  2013.

\bibitem[Karwa \& Slavkovi\'c(2016)Karwa and Slavkovi\'c]{Karwa2016:Inference}
Karwa, V. and Slavkovi\'c, A.
\newblock Inference using noisy degrees: Differentially private $\beta$-model
  and synthetic graphs.
\newblock \emph{The Annals of Statistics}, 44\penalty0 (1):\penalty0 87--112,
  02 2016.
\newblock \doi{10.1214/15-AOS1358}.

\bibitem[Karwa et~al.(2016)Karwa, Krivitsky, and
  Slavkovi\'c]{Karwa2016:Sharing}
Karwa, V., Krivitsky, P.~N., and Slavkovi\'c, A.~B.
\newblock Sharing social network data: differentially private estimation of
  exponential family random‐graph models.
\newblock \emph{Journal of the Royal Statistical Society: Series C (Applied
  Statistics)}, 66\penalty0 (3):\penalty0 481--500, 2016.
\newblock \doi{10.1111/rssc.12185}.

\bibitem[Kifer \& Lin(2010)Kifer and Lin]{Kifer2010}
Kifer, D. and Lin, B.-R.
\newblock Towards an axiomatization of statistical privacy and utility.
\newblock In \emph{Proceedings of the twenty-ninth ACM SIGMOD-SIGACT-SIGART
  symposium on Principles of database systems}, pp.\  147--158. ACM, 2010.

\bibitem[Kifer et~al.(2012)Kifer, Smith, and Thakurta]{Kifer2012:PrivateCERM}
Kifer, D., Smith, A., and Thakurta, A.
\newblock Private convex empirical risk minimization and high-dimensional
  regression.
\newblock \emph{Journal of Machine Learning Research}, 1:\penalty0 1--41, 01
  2012.

\bibitem[Kokoszka \& Reimherr(2017)Kokoszka and
  Reimherr]{kokoszka2017introduction}
Kokoszka, P. and Reimherr, M.
\newblock \emph{Introduction to functional data analysis}.
\newblock CRC Press, 2017.

\bibitem[McSherry \& Talwar(2007)McSherry and Talwar]{McSherry2007}
McSherry, F. and Talwar, K.
\newblock Mechanism design via differential privacy.
\newblock In \emph{Proceedings of the 48th Annual IEEE Symposium on Foundations
  of Computer Science}, FOCS '07, pp.\  94--103, Washington, DC, USA, 2007.
  IEEE Computer Society.
\newblock ISBN 0-7695-3010-9.
\newblock \doi{10.1109/FOCS.2007.41}.

\bibitem[{Mirshani} et~al.(2017){Mirshani}, {Reimherr}, and
  {Slavkovic}]{Mirshani2007}
{Mirshani}, A., {Reimherr}, M., and {Slavkovic}, A.
\newblock {On the Existence of Densities for Functional Data and their Link to
  Statistical Privacy}.
\newblock \emph{ArXiv e-prints}, November 2017.

\bibitem[Ramsay \& Silverman(2005)Ramsay and Silverman]{ramsay2005springer}
Ramsay, J. and Silverman, B.
\newblock \emph{Functional Data Analysis}.
\newblock Springer, 2005.

\bibitem[Ramsay et~al.(2018)Ramsay, Wickham, Graves, and Hooker]{fdapackage}
Ramsay, J.~O., Wickham, H., Graves, S., and Hooker, G.
\newblock \emph{fda: Functional Data Analysis}, 2018.
\newblock R package version 2.4.8.

\bibitem[Sheffet(2017)]{Sheffet2017}
Sheffet, O.
\newblock Differentially private ordinary least squares.
\newblock In Precup, D. and Teh, Y.~W. (eds.), \emph{Proceedings of the 34th
  International Conference on Machine Learning}, volume~70 of \emph{Proceedings
  of Machine Learning Research}, pp.\  3105--3114, International Convention
  Centre, Sydney, Australia, 06--11 Aug 2017. PMLR.

\bibitem[Smith(2011)]{Smith2011:Privacy-preservingSE}
Smith, A.
\newblock Privacy-preserving statistical estimation with optimal convergence
  rates.
\newblock In \emph{Proceedings of the Forty-third Annual ACM Symposium on
  Theory of Computing}, STOC '11, pp.\  813--822, New York, NY, USA, 2011. ACM.
\newblock ISBN 978-1-4503-0691-1.
\newblock \doi{10.1145/1993636.1993743}.

\bibitem[Smith et~al.(2018)Smith, Álvarez, Zwiessele, and Lawrence]{Smith2018}
Smith, M., Álvarez, M., Zwiessele, M., and Lawrence, N.~D.
\newblock Differentially private regression with gaussian processes.
\newblock In Storkey, A. and Perez-Cruz, F. (eds.), \emph{Proceedings of the
  Twenty-First International Conference on Artificial Intelligence and
  Statistics}, volume~84 of \emph{Proceedings of Machine Learning Research},
  pp.\  1195--1203, Playa Blanca, Lanzarote, Canary Islands, 09--11 Apr 2018.
  PMLR.

\bibitem[Vu \& Slavkovic(2009)Vu and Slavkovic]{Vu2009}
Vu, D. and Slavkovic, A.
\newblock Differential privacy for clinical trial data: Preliminary
  evaluations.
\newblock In \emph{Proceedings of the 2009 IEEE International Conference on
  Data Mining Workshops}, ICDMW '09, pp.\  138--143, Washington, DC, USA, 2009.
  IEEE Computer Society.
\newblock ISBN 978-0-7695-3902-7.
\newblock \doi{10.1109/ICDMW.2009.52}.

\bibitem[Wang et~al.(2018)Wang, Kifer, Lee, and Karwa]{Wang2018}
Wang, Y., Kifer, D., Lee, J., and Karwa, V.
\newblock Statistical approximating distributions under differential privacy.
\newblock \emph{Journal of Privacy and Confidentiality}, 8\penalty0 (1), 2018.

\bibitem[Wang et~al.(2015)Wang, Fienberg, and Smola]{Wang2015:PrivacyFree}
Wang, Y.-X., Fienberg, S.~E., and Smola, A.~J.
\newblock Privacy for free: Posterior sampling and stochastic gradient monte
  carlo.
\newblock In \emph{Proceedings of the 32nd International Conference on
  International Conference on Machine Learning - Volume 37}, ICML'15, pp.\
  2493--2502. JMLR.org, 2015.

\bibitem[Wasserman \& Zhou(2010)Wasserman and
  Zhou]{Wasserman2010:StatisticalFDP}
Wasserman, L. and Zhou, S.
\newblock A statistical framework for differential privacy.
\newblock \emph{Journal of the American Statistical Association},
  105:489:\penalty0 375--389, 2010.

\bibitem[Yu et~al.(2014)Yu, Rybar, Uhler, and Fienberg]{Yu2014}
Yu, F., Rybar, M., Uhler, C., and Fienberg, S.~E.
\newblock Differentially-private logistic regression for detecting multiple-snp
  association in gwas databases.
\newblock In \emph{Privacy in Statistical Databases: UNESCO Chair in Data
  Privacy, International Conference, PSD 2014, Ibiza, Spain, September 17-19,
  2014. Proceedings}, pp.\  170--184, Cham, 2014. Springer International
  Publishing.
\newblock ISBN 978-3-319-11257-2.
\newblock \doi{10.1007/978-3-319-11257-2_14}.

\bibitem[Zhang et~al.(2012)Zhang, Zhang, Xiao, Yang, and
  Winslett]{Zhang2012:FunctionalMR}
Zhang, J., Zhang, Z., Xiao, X., Yang, Y., and Winslett, M.
\newblock Functional mechanism: Regression analysis under differential privacy.
\newblock \emph{Proc. VLDB Endow.}, 5\penalty0 (11):\penalty0 1364--1375, July
  2012.
\newblock ISSN 2150-8097.
\newblock \doi{10.14778/2350229.2350253}.

\end{thebibliography}

\end{document}